%% file: root.tex
\documentclass[letterpaper, 10 pt, conference]{ieeeconf}  

\IEEEoverridecommandlockouts                              

\usepackage[scaled=.92]{helvet}
\usepackage{times}
\usepackage{graphicx}
\usepackage{amsmath,amssymb}
\usepackage{amsthm}
\usepackage{bm}
\usepackage{parskip}
\usepackage{color}
\usepackage[center]{subfigure}

\title{\LARGE \bf
A Penetration Metric for Deforming Tetrahedra using Object Norm
}

\author{Jisu Kim$^{\dagger}$ and Young J. Kim$^{\dagger}$
\thanks{$^\dagger$J. Kim and Y. J. Kim are with the Department of Computer Science and Engineering at Ewha Womans University in Korea  {\tt\small cwyh5526@ewhain.net,kimy@ewha.ac.kr}}
}

\theoremstyle{plain}
\newtheorem{theorem}{Theorem}

\newtheorem{definition}{Definition} 
\newcommand{\PDd}{\mathrm{PD}_d}
\newcommand{\Ta}{\mathcal{T}_1} 
\newcommand{\Tb}{\mathcal{T}_2} 

\DeclareMathOperator*{\argmin}{argmin}

\renewenvironment{thebibliography}[1]{%
   \begin{oldthebibliography}{#1}%
     \setlength{\itemsep}{-1.8ex}%
}%
{%
   \end{oldthebibliography}%
}
\begin{document}

\maketitle
\thispagestyle{empty}
\pagestyle{empty}

\begin{abstract}
In this paper, we propose a novel penetration metric, called {\it deformable penetration depth} $\PDd$, to define a measure of inter-penetration between two linearly deforming tetrahedra using the object norm \cite{ObjectNorm}.  First of all, we show that a distance metric for a tetrahedron deforming between two configurations can be found in closed form based on object norm. Then, we show that the $\PDd$ between an intersecting pair of static and deforming tetrahedra can be found by solving a quadratic programming (QP) problem in terms of the distance metric with non-penetration constraints. We also show that the $\PDd$ between two, intersected, deforming tetrahedra can be found by solving a similar QP problem under some assumption on penetrating directions, and it can be also accelerated by an order of magnitude using pre-calculated penetration direction.  We have implemented our algorithm on a standard PC platform using an off-the-shelf QP optimizer, and experimentally show that both the static/deformable and deformable/deformable tetrahedra cases can be solvable in from a few to tens of milliseconds. Finally, we demonstrate that our penetration metric is three-times smaller (or tighter) than the classical, rigid penetration depth metric in our experiments.
\end{abstract}

\section{INTRODUCTION}

\vspace{-1mm}
Handling soft objects is increasingly popular and becoming important in the robotics community. For instance, soft object manipulation is gaining broad attention recently due to robustness in manipulation and practical real-world demand \cite{kim2013soft}. In bio-inspired robotics, compliant actuation is a must to mimic that of biological entities. In all these areas,  accurately simulating deformation for soft objects possibly with intensive contact is very crucial to ensure the robustness of robots and reduce their building cost.

\vspace{-1mm}
The finite element method (FEM) is a general method to simulate deformable motion for soft objects with diverse material and structural properties, and it has been extensively studied for many decades in the area of computer-aided engineering (CAE), structural dynamics, computer animation and soft robotics \cite{rao2017finite,sifakis2012fem}. Typically, the FEM models a soft object with a network of many finite elements (FEs), such as tetrahedra. In order to accurately simulate contact dynamics using FEM, it is important to define a proper penetration measure between FEs and apply proper responsive forces to them.

\vspace{-1mm}
However, unfortunately, there has been no rigorous formulation in the literature to define a penetration measure or metric for deformable FEs, in particular tetrahedra. In this paper, we propose a novel penetration metric, called {\it deformable penetration depth} $\PDd$, to define a measure of inter-penetration between two linearly deforming tetrahedra using the object norm \cite{ObjectNorm}.
The main contribution of our paper can be summarized as follows:
\begin{itemize}
    \item In order to define $\PDd$, we generalize the concept of object norm  to a deformable case, and show that a distance metric for a tetrahedron deforming between two configurations can be found in closed form.
    
    \item  Case 1: we show that the $\PDd$ between an intersected pair of one static and one deforming tetrahedra can be found by solving a quadratic programming (QP) problem over all possible penetrating directions in terms of the distance metric with non-penetration constraints. 
    
    \item Case 2: we show that the $\PDd$ between an intersected pair of deforming tetrahedra can be found by solving a similar QP problem.
    
    
    \item  Case 3: We approximate the case 2 using rigid penetration depth calculation based on the separating axis theorem, and demonstrate that this computation can be accelerated by an order of magnitude compared to the case 2 while the approximation error is kept to less than by 5\%.
    
    \item We compare the metric results of our deformable penetration depth against those of rigid penetration depth and show that our results are three-times tighter than the rigid case.

\end{itemize}

The rest of this paper is organized as follows. In the next section, we briefly survey works relevant to penetration depth computation. In Sec.~\ref{sec:objnorm}, we show a closed-form solution of our penetration metric and subsequently apply it to obtain $\PDd$ for static vs. deformable (Sec.~\ref{sec:case1}) and deformable vs. deformable tetrahedral cases (Sec.~\ref{sec:case2}). In Sec.~\ref{sec:case3}, we show that these solutions can be obtained more rapidly. In Sec.~\ref{sec:results}, we show various experimental results of our penetration depth algorithms and conclude the paper in Sec.~\ref{sec:conc}.


\section{RELATED WORKS}

\subsection{Penetration Metric for Rigid Objects}
Various penetration metrics for rigid objects have been suggested in the literature \cite{LMK17}. The (translational) penetration depth $\mathrm{PD}$ is the most well-known metric, defined by a minimal distance that separates two intersecting objects \cite{cameron1986determining, dobkin1993computing}. 
Approximation methods have been proposed for convex models \cite{van2001proximity, kim2004incremental}, and for non-convex models, exact $\mathrm{PD}$ \cite{Hachenberger09} and approximations \cite{kim2002fast, RedonLin06, Lien08}, \cite{lien09}, and real-time algorithms \cite{PolyDepth} are suggested. Kim et al. \cite{kim2015hybrid} presented a hybrid algorithm that combines local optimization with machine learning.

The continuity of penetration metric has been considered; the Phong projection \cite{lee2015phongpd} and dynamic Minkowski sums \cite{lee2016continuous} were used to compute continuous $\mathrm{PD}$. Both positive and negative distances (i.e., PD), and the minimal distance value over a continuous time interval can be computed \cite{lee2019contdist}.
The generalized penetration depth $\mathrm{PD}_{g}$ is defined as a minimal rigid transformation that interpenetrating objects must undergo to resolve the penetration \cite{zhang2007generalized}.
%
The works by \cite{Weller-RSS-09, Allard2010, wang2012adaptive} use penetration volume as a continuous penetration metric. Nirel and Lischinski \cite{nirel2018fast} also presented a volume-based global collision resolution method. The pointwise penetration depth is defined as a distance between the deepest intersecting points of two objects \cite{tang2009interactive}.
\vspace{-1.5mm}

\subsection{Penetration Metric for Deformable Objects}

\vspace{-1mm}
Computing a penetration metric (or penetration depth) for deformable objects has been relatively less studied than for rigid objects. In particular, no attempt has been made to rigorously define such a metric, and our work tackles this problem.  

\vspace{-1mm}
Distance-fields-based representations are often used  to compute penalty forces for deformable objects \cite{fisher2001fast} and can be accelerated using GPUs \cite{hoff2002fast, sud2006fpc}.
Heidelberger et al. \cite{heidelberger2004consistent} improved the consistency of penetration depth using a propagation scheme.
There are few FEM-based methods that calculate penetration metric.
\cite{hirota2001implicit} suggested material depth, an approximation of distance field, which is the distance between the interior and object-boundary points. 
The energy-based method \cite{teran2005robust} also used a signed distance in material space to compute the penetration depth.
Layered depth images (LDIs) \cite{Shade:1998:LDI:280814.280882} are a data structure representing multiple layers of geometry rendered from a fixed viewpoint. Heidelberger et al. \cite{heidelberger2003real, heidelberger2004detection} suggested a method to estimate penetration volume using pixel depths and normals from LDIs.  
The methods by \cite{faure2012sofa} used a surface rasterization method based on LDIs \cite{faure2008image, Allard2010} to compute repulsive forces. 

\vspace{-1mm}
However, most of these existing works are heuristically defined and more seriously, they do not guarantee full separation of intersecting objects. In other words, all these metrics provide only a lower bound of penetration depth, not an upper bound, unlike our method. 

\section{Metric Formulation}\label{sec:objnorm}

\vspace{-1mm}
Given a tetrahedron $\mathcal{T} \in \mathbb{R}^3$ linearly deforming between rest $\mathbf{q}_0 \in \mathbb{R}^{12}$ and deformed configurations $\mathbf{q_1} \in \mathbb{R}^{12}$, 
an arbitrary point $\mathbf{r} \in \mathcal{T}(\mathbf{q_0})$ and its counterpart $\mathbf{p} \in \mathcal{T}(
\mathbf{q_1})$ can be represented using barycentric coordinate $\mathbf{b}=(b_1, b_2, b_3)$ which is constant during deformation \cite{sifakis2012fem} (also illustrated in Fig.~\ref{fig:linearDeformation}):
\vspace{-1mm}
\begin{equation}\nonumber
\begin{split}
    \mathbf{r} &= \mathbf{r}_0+(\mathbf{r}_1-\mathbf{r}_0)b_1 +(\mathbf{r}_2-\mathbf{r}_0)b_2+(\mathbf{r}_3-\mathbf{r}_0)b_3\\
    \mathbf{p} &= \mathbf{p}_0+(\mathbf{p}_1-\mathbf{p}_0)b_1 +(\mathbf{p}_2-\mathbf{p}_0)b_2+(\mathbf{p}_3-\mathbf{p}_0)b_3.
\end{split}
\end{equation}

\vspace{-1.5mm}
Then, the distance metric $\sigma(\mathbf{q_0, q_1})$ using object norm \cite{ObjectNorm,fap} can be formulated as: 
\begin{equation}\label{eq:objdef}
\sigma(\mathbf{q_0, q_1})=\frac{1}{V} \int_{} \left\Vert \mathbf{p}-\mathbf{r} \right\Vert^2 dV
={6}\!\int\! \left\Vert {\mathbf{D}\mathbf{b}+\mathbf{d}_0}\right\Vert ^2 d\mathbf{b},
\end{equation}

\begin{figure}[htb]
      \centering\subfigure[Rest $\mathbf{q}_0$] {\includegraphics[height=2.5cm]{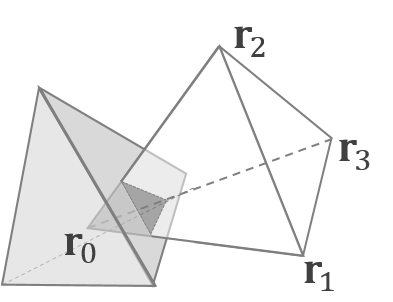}}
       \centering\subfigure[Deformed $\mathbf{q}_1$] {\includegraphics[height=2.5cm]{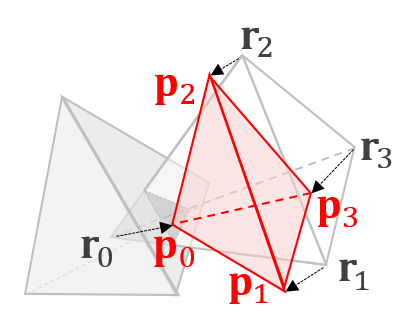}}
       \vspace{-1.5mm}
      \caption{Linear deformation of a tetrahedron to resolve  inter-penetration from the rest $\mathbf{q}_0$ to deformed configuration $\mathbf{q}_1$.}
      \label{fig:linearDeformation}
  \end{figure}
  
where $V=\frac{1}{6}$ is the volume of $\mathcal{T}$ in terms of the barycentric coordinate,  $\mathbf{d}_i=(\mathbf{p}_i-\mathbf{r}_i)$,  
and $\mathbf{D} \in \mathbb{R}^{3\times3}=[\mathbf{d}_1-\mathbf{d}_0 \vert \mathbf{d}_2-\mathbf{d}_0 \vert \mathbf{d}_3-\mathbf{d}_0 ]$. Eq.~\ref{eq:objdef} induces a new tetrahedron $\mathcal{T}_d$, called {\it displacement tetrahedron}, with four vertices $\mathbf{d}_i=(\mathbf{p}_i-\mathbf{r}_i), i=0, \cdots, 3$, and the object norm is a sum of squared distances from $\mathbf{d}_0$ for $\forall \mathbf{d} \in \mathcal{T}_d$.


\vspace{-1mm}
Let $\mathbf{d}_{i}=(x_i,y_i,z_i)^T$. Then the closed-form solution of the proposed metric can be calculated as: \
\begin{equation}\label{eq:closed}
\sigma(\mathbf{q_0, q_1})
=\frac{1}{10} \sum_{\substack{\forall i\geq j\in\{0,\dots,3\}}}\mathbf{d}_{i}^T \mathbf{d}_{j}\\
\end{equation}
Let $\mathbf{I}_d$ be the inertia tensor of the displacement tetrahedron $\mathcal{T}_d$. Then, the moments of inertia of $\mathbf{I}_d$, $\mathbf{I}_{xx}, \mathbf{I}_{yy}, \mathbf{I}_{zz}$, with {\it constant} mass density $\mu_d$ are:
\begin{equation}\nonumber
\begin{aligned}
\mathbf{I}_{xx} &=& \frac{1}{60} \mu_d \vert det(\mathbf{D}) \vert  \sum_{\forall i\geq j}\mathbf{d}_{i}^T [\mathbf{0} \; \mathbf{e}_2 \; \mathbf{e}_3 ] \mathbf{d}_{j}\\
\mathbf{I}_{yy} &=& \frac{1}{60} \mu_d \vert det(\mathbf{D}) \vert  \sum_{\forall i\geq j}\mathbf{d}_{i}^T [\mathbf{e}_1 \; \mathbf{0} \; \mathbf{e}_3 ] \mathbf{d}_{j}\\
\mathbf{I}_{zz} &=& \frac{1}{60} \mu_d \vert det(\mathbf{D}) \vert  \sum_{\forall i\geq j}\mathbf{d}_{i}^T [\mathbf{e}_1 \; \mathbf{e}_2 \; \mathbf{0} ] \mathbf{d}_{j},
\end{aligned}
\end{equation}
where $\mathbf{e}_i$ is the standard basis for $\mathbb{R}^3$.
Then, since $\vert det(\mathbf{D}) \vert = 6V_d = 1$ where $V_d$ is the volume of $\mathcal{T}_d$ (which is always $\frac{1}{6}$), 
\begin{equation}\label{eq:objnorm}
\begin{split}
\sigma(\mathbf{q_0, q_1})&= \frac{3}{\mu_d }(\mathbf{I}_{xx}+\mathbf{I}_{yy}+\mathbf{I}_{zz})=\frac{3}{\mu_d}tr(\mathbf{I_d})
\end{split}
\end{equation}

\vspace{-1mm}
{\bf Definition of Deformable PD}
We use Eq.~\ref{eq:closed} or \ref{eq:objnorm} to define the deformable penetration depth $\PDd$ as follows. 
Given an intersecting pair of linearly deforming tetrahedra, $\Ta, \Tb$, with the corresponding, initial configurations $\mathbf{q}_0^{\Ta}, \mathbf{q}_0^{\Tb}$, their $\PDd$ is:
\begin{equation}\label{eq:d_PD_definition}
\begin{split}
&\PDd(\Ta(\mathbf{q}_0^{\Ta}), \Tb(\mathbf{q}_0^{\Tb}))\\
&=\min_{(\mathbf{q}_1^{\Ta},\mathbf{q}_1^{\Tb})\in \mathcal{C}} \sqrt{\sigma(\mathbf{q}_0^{\Ta}, \mathbf{q}_1^{\Ta}) +
\sigma(\mathbf{q}_0^{\Tb}, \mathbf{q}_1^{\Tb})}
\end{split}
\end{equation}
where $\mathcal{C} \subset \mathbb{R}^{24}$ is the contact (configuration) space imposed by $\Ta, \Tb$. In Sec.~\ref{sec:case1}, we present our $\PDd$ computation algorithm when $\Ta$ is static; i.e.,  $\mathbf{q}_0^{\Ta}=\mathbf{q}_1^{\Ta}$, and relax this restriction in Sec.~\ref{sec:case2} by allowing both $\Ta, \Tb$ deformable.

{\bf Generalization of Rigid PD}
Note that the proposed $\PDd$ is a generalization of the classical rigid (or translational) penetration depth \cite{cameron1986determining, dobkin1993computing}.
Since the object norm in Eq.~\ref{eq:objnorm} can be interpreted as an average displacement of all the points inside a deformable object during deformation, when the object purely translates by $\mathbf{d}$, the object norm will be equivalent to the squared norm of the translation vector $\mathbf{d}$ from Eq.~\ref{eq:closed}: i.e.,
\begin{equation}
\sigma = \frac{1}{10} \sum{\mathbf{d}_{i}^{T} \mathbf{d}_{i}} = {\mathbf{d}_{i}^{T}\mathbf{d}_{i}} = \Vert \mathbf{d}\Vert^2 
\end{equation}
Thus, according to Eq.~\ref{eq:d_PD_definition}, $\PDd=\mathrm{PD}$ when $\mathbf{q} \in \mathrm{SE(3)}$.




\section{Static vs Deforming Tetrahedra}\label{sec:case1}

\vspace{-1mm}
In this section, we solve the optimization problem of Eq.~\ref{eq:d_PD_definition} when $\Ta$ is static but $\Tb$ is still deformable.
Since the contact space $\mathcal{C}$ in Eq.~\ref{eq:d_PD_definition} can be realized by  face-vertex (FV), vertex-face (VF) or edge-edge (EE) contacts (Fig.~\ref{fig:contacts}),  
the normal vector of the contact plane is used as a direction to separate $\Ta, \Tb$, which can also minimize Eq.~\ref{eq:d_PD_definition}. According to the separating axis theorem \cite{gottschalk1996obbtree}, the total number of possible separating directions that we need to consider is 44 since there are 8 VF (or FV) and 36 EE contact pairs between two tetrahedra. 

\vspace{-1mm}
In general, when both $\Ta, \Tb$ are deforming, even for linear deformation, the separating directions can be non-linear, which makes Eq.~\ref{eq:d_PD_definition} hard to solve. In our work, to make the problem tractable, we assume that the separating directions are always obtained from the rest configurations $\mathbf{q}_0^{\Ta}, \mathbf{q}_0^{\Tb}$ of $\Ta, \Tb$, thus constant. Note that this is a reasonable assumption unless objects deform severely, which is reasonable for most practical robotic applications.
Let us call a vertex in $\Ta$ or $\Tb$ {\em constrained} if they are involved in calculating the separating direction (i.e., the contact normal); otherwise, call it {\em free}. For example, if a face normal of $\Ta$ (e.g., $\mathbf{n}_\mathrm{FV}$ in Fig.~\ref{fig:contacts}(a)) is selected as a separating direction, the vertices incident to the face (e.g., $\mathbf{s}_{0},\mathbf{s}_{1}, \mathbf{s}_{2}$, in Fig.~\ref{fig:contacts}(a)) are constrained. After optimization, free vertices may or may not be on the contact plane, but constrained vertices are always on the plane. 

\vspace{-1mm}
\subsection{Separating Direction Calculation}\label{sec:sepdir}
\vspace{-1mm}
Let $\{\mathbf{s}_{i} \} \subset \Ta, \{ \mathbf{r}_{i} \} \subset \Tb(\mathbf{q}_{0}^{\Tb}), \{ \mathbf{p}_{i} \} \subset \Tb(\mathbf{q}_{1}^{\Tb}) , i=0\sim3$ be the four vertices of $\Ta, \Tb(\mathbf{q}_{0}^{\Tb}), \Tb(\mathbf{q}_{1}^{\Tb})$, respectively. We first calculate the constant normal vector of a separating plane using the rest configurations. Let $\mathbf{n}$ be such a  normal vector for each contact state. Then the following is the result of normal vector calculation:
\begin{equation}\label{eq:new_normals}
\begin{aligned}
\mathbf{n}_{\mathrm{FV}}&=(\mathbf{s}_{1} - \mathbf{s}_{0})\times(\mathbf{s}_{2} - \mathbf{s}_{0})\\
\mathbf{n}_{\mathrm{VF}}&=(\mathbf{r}_{1} - \mathbf{r}_{0})\times(\mathbf{r}_{2} - \mathbf{r}_{0})\\
\mathbf{n}_{\mathrm{EE}}&=(\mathbf{r}_{1} - \mathbf{r}_{0})\times(\mathbf{s}_{1} - \mathbf{s}_{0})
\end{aligned}
\end{equation}
Furthermore, to make the separation direction consistent, the direction is decided based on the following rules:
\begin{enumerate}
\item $\mathbf{n}_{\mathrm{FV}}$ should be outward from $\Ta$.
\item $\mathbf{n}_{\mathrm{VF}}$ should be inward to $\Tb$.
\item $\mathbf{n}_{\mathrm{EE}}$ should point away from the non-contacting vertices of $\Ta$  (e.g., $\mathbf{s}_{2}, \mathbf{s}_{3}$ in Fig.~\ref{fig:contacts}(c)).
\end{enumerate}
 \begin{figure}[htb]
      \centering\subfigure[FV contact] {\includegraphics[height=2.2cm]{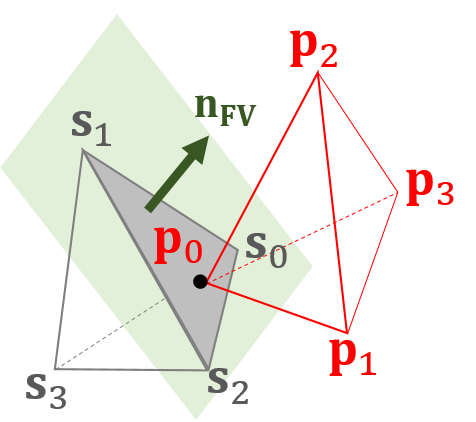}}
      \centering\subfigure[VF contact]{\includegraphics[height=2.2cm]{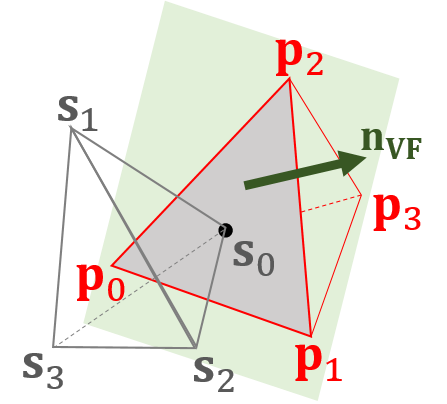}}
      \centering\subfigure[EE contact]{\includegraphics[height=2.2cm]{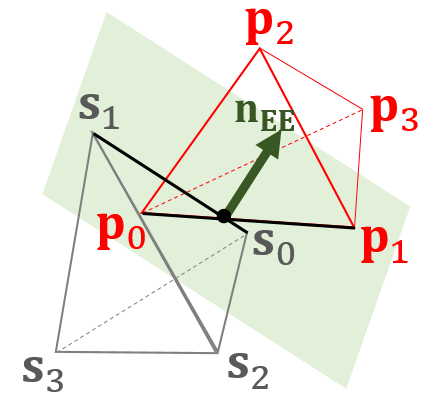}}
      \vspace{-1.5mm}
      \caption{Three contact cases (FV, VF, EE) after penetration resolution. $\{\mathbf{s}_i\} \subset \Ta$ is  static and $\{\mathbf{p}_i\} \subset \Tb$ is a deforming tetrahedron}
      \label{fig:contacts}
  \end{figure}
  
 \begin{figure}[htb]
      \centering\subfigure[Normal] {\includegraphics[height=2.2cm]{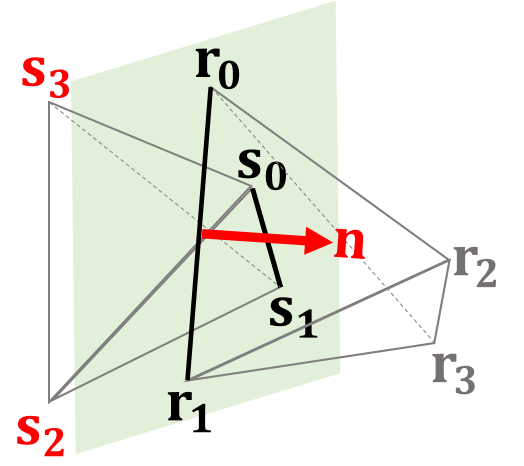}}
      \centering\subfigure[Undecidable] {\includegraphics[height=2.2cm]{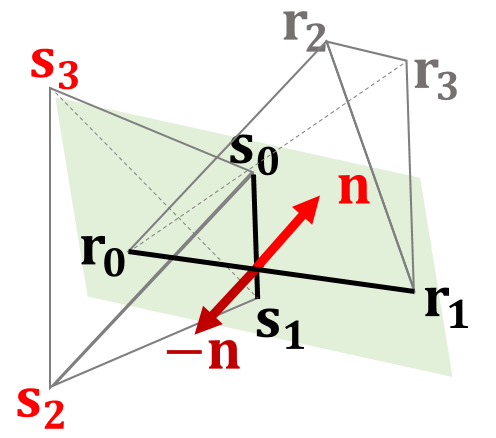}}
      \centering\subfigure[Parallel]
      {\includegraphics[height=2.2cm]{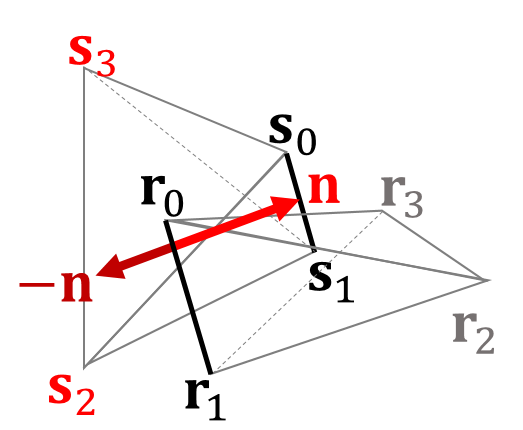}}
      \vspace{-1.5mm}
      \caption{Various cases to deal with the separating direction $\mathbf{n}$ for an EE contact. }
      \vspace{-3mm}
      \label{fig:normals}
   \end{figure}
\vspace{-1mm}
Note that the third case for $\mathbf{n}_{\mathrm{EE}}$ works only when the non-contacting vertices of $\Ta$ lie in the same half space  (Fig.~\ref{fig:normals}(a)). If this is not the case (Fig.~\ref{fig:normals}(b)), both directions are tested to minimize Eq.~\ref{eq:d_PD_definition}.  There also exist degenerate cases for $\mathbf{n}_{\mathrm{EE}}$ when two contacting EE pairs are parallel to each other. In this case, the normal is calculated from the shortest distance vector between the two EE pairs (Fig.~\ref{fig:normals}(c)). When the two EE pairs are co-linear, the separating direction can be any of the vectors perpendicular to the edge. Thus, we can use the normal directions of the faces incident to the edges as a candidate separating direction, which is redundant as it can be covered by VF or FV cases.


 \vspace{-1mm}
\subsection{Non-penetration Constraints}\label{sec:case1non}
 \vspace{-1mm}
Eq.~\ref{eq:d_PD_definition} is essentially a constrained optimization problem where the constraints are non-penetration constraints. Due to the constant assumption on the separating directions employed in Sec.~\ref{sec:sepdir}, now we can linearize the non-penetration constraints and thus set up a solvable QP problem afterwards.
Specifically, we write the non-penetration constraints for each contact case as follows: 
 \vspace{-1mm}
\begin{enumerate}
    \item FV case (Fig.~\ref{fig:contacts}(a)): since all the vertices $\mathbf{s}_i$ of the static tetrahedron $\Ta$ lie in the same half space, we simply impose non-penetration constraints for the vertices $\mathbf{p}_i$ of the deforming tetrahedron $\Tb$ as follows, where $\mathbf{s}_{0}$ is the vertex defining the FV contact:
\begin{equation}\label{eq:new_constraint1}
\begin{aligned}
\mathbf{n}_{\mathrm{FV}}\cdot(\mathbf{p}_{i}-\mathbf{s}_{0}) &\geq 0
\end{aligned}
\end{equation}
\item VF case (Fig.~\ref{fig:contacts}(b)): with constrained vertices ($\mathbf{p}_{k}, k = 0\sim2$) of $\Tb$ and free vertices ($\mathbf{s}_{j}, j = 0\sim3$) of $\Ta$, the free vertex $\mathbf{p}_{3}$ and the constrained vertex $\mathbf{p}_{0}$ on the plane, we formulate seven non-penetration constraints:
\begin{equation}\label{eq:new_constraint2}
\begin{aligned}
\mathbf{n}_{\mathrm{VF}}\cdot(\mathbf{p}_{k}-\mathbf{p}_{0}) = 0, \\
\mathbf{n}_{\mathrm{VF}}\cdot(\mathbf{s}_{j}-\mathbf{p}_{0}) \leq 0, \\
\mathbf{n}_{\mathrm{VF}}\cdot(\mathbf{p}_{3}-\mathbf{p}_{0}) \geq 0,
\end{aligned}
\end{equation}
where $k \in \{0,\dots,2\}, j \in \{0,\dots,3\}$.

\item EE case (Fig.~\ref{fig:contacts}(c)): with constrained vertices ($\mathbf{s}_{0}, \mathbf{s}_{1}, \mathbf{p}_{0}, \mathbf{p}_{1}$),  the free vertices ($\mathbf{s}_{2}, \mathbf{s}_{3}$) should lie in the same half space, and free vertices ($\mathbf{p}_{2}, \mathbf{p}_{3}$) should lie in the other half space. Without loss of generality, assuming that $\mathbf{p}_{0}$ is on the separating plane, we can formulate seven constraints.
\begin{equation}\label{eq:new_constraint3}
\begin{aligned}
\mathbf{n}_{\mathrm{EE}}\cdot(\mathbf{s}_{k}-\mathbf{p}_{0}) = 0\\
\mathbf{n}_{\mathrm{EE}}\cdot(\mathbf{p}_{1}-\mathbf{p}_{0}) = 0\\
\mathbf{n}_{\mathrm{EE}}\cdot(\mathbf{s}_{j}-\mathbf{p}_{0}) \leq 0\\
\mathbf{n}_{\mathrm{EE}}\cdot(\mathbf{p}_{j}-\mathbf{p}_{0}) \geq 0
\end{aligned}
\end{equation}
where $k \in \{0,1\}, j \in \{2,3\}$.
\end{enumerate}


\vspace{-1.5mm}
\subsection{Constrained Optimization}
 \vspace{-1mm}
Since now we set up both an objective function (Eq.~\ref{eq:objnorm}) and non-penetration constraints for each contact case (Eqs.~\ref{eq:new_constraint1},\ref{eq:new_constraint2},\ref{eq:new_constraint3}), we solve the constrained optimization in Eq.~\ref{eq:d_PD_definition}. Since the objective function is quadratic and the constraints are linear, our problem is a QP problem. This problem is solvable as the objective function is semi-positive definite and the constraint space is convex.
We use an off-the-shelf QP solver like Gurobi Optimizer \cite{gurobi} to effectively solve this problem.


\section{Deforming vs Deforming Tetrahedra}\label{sec:case2}
 \vspace{-1mm}
In this section, we extend the optimization problem of Eq.~\ref{eq:d_PD_definition} to the case where both $\Ta$ and $\Tb$ are deforming. 
%
 \vspace{-1mm}
\subsection{Non-penetration Constraints}
 \vspace{-1mm}
The non-penetration constraints used for this section requires a new variable ${\mathbf{p}_{s}} \in \mathbb{R}^{3}$, which is a point on the separating plane to decide the position of the plane. For the static and deforming tetrahedron case, the position of a separating plane is automatically decided once we determine the separating direction since there will be at least one vertex of the static tetrahedron involved in contact. But now, not only the configurations of both tetrahedra but also the position of the separating plane should be optimized at the same time because there is no static vertex to decide the position.
By introducing ${\mathbf{p}_{s}}$, it becomes easier to write the constraints since we do not have to consider which vertex to select as a point on the plane. The only difference between each contact case is the separating direction and a set of constrained vertices, $\mathcal{C}_s$ for $\Ta$ and $\mathcal{C}_p$ for $\Tb$; for instance, $\mathcal{C}_s=\{\mathbf{s}_{0},\mathbf{s}_{1},\mathbf{s}_{2}\}, \mathcal{C}_p=\emptyset$ in Fig.~\ref{fig:constraints}(a) and $\mathcal{C}_s=\{\mathbf{s}_{0},\mathbf{s}_{1} \}, \mathcal{C}_p=\{\mathbf{p}_{0},\mathbf{p}_{1}\} $ in Fig.~\ref{fig:constraints}(b). 
Once they are decided,  we can write the non-penetration constraints in a more general form than those presented in Sec.~\ref{sec:case1}. With a chosen separating direction $\mathbf{n}$ for each contact case, 
the corresponding constrained vertices $\mathbf{s}_{{c}} \in \Ta$, $\mathbf{p}_{{c}} \in \Tb$ and the free vertices $\mathbf{s}_{{f}} \in \Ta$, $\mathbf{p}_{{f}} \in \Tb$ of each deforming tetrahedron $\Ta, \Tb$,
the eight non-penetration constraints can be formulated as:
\begin{equation}\label{eq:constraints}
\begin{aligned}
\mathbf{n}\cdot(\mathbf{s}_{{c}}-\mathbf{p}_{{s}}) &= 0, ~\forall \mathbf{s}_{c} \in \mathcal{C}_{s},\\
\mathbf{n}\cdot(\mathbf{p}_{{c}}-\mathbf{p}_{{s}}) &= 0, ~\forall \mathbf{p}_{c} \in \mathcal{C}_{p},\\
\mathbf{n}\cdot(\mathbf{s}_{{f}}-\mathbf{p}_{{s}}) &\leq 0, ~\forall \mathbf{s}_{f} \in \mathcal{F}_{s},\\
\mathbf{n}\cdot(\mathbf{p}_{{f}}-\mathbf{p}_{{s}}) &\geq 0, ~\forall \mathbf{p}_{f} \in \mathcal{F}_{p},\\
\end{aligned}
\end{equation}
where $\mathcal{F}_{s}, \mathcal{F}_{p}$ are the set of free vertices for $\Ta, \Tb$, respectively.

 \vspace{-1mm}
The non-penetration constraints for the static/deforming tetrahedron in Sec.~\ref{sec:case1} can be viewed as a special case of this form where ${\mathbf{s}_{{c}}}$ and ${\mathbf{s}_{f}}$ are fixed and ${\mathbf{p}_{s}}$ is chosen among constrained vertices. 
 For example, in FV case, since $\mathcal{C}_{p} = \emptyset$, the second constraint in Eq.~\ref{eq:constraints} can be ignored. Moreover, as ${\mathbf{p}_{s}}$ is chosen from ${\mathcal{C}_{s}}$, the first and third constraints in Eq.~\ref{eq:constraints} are automatically satisfied and the rest of four constraints remain just like in Eq.~\ref{eq:new_constraint1}. The VF case (Eq.~\ref{eq:new_constraint2}) and EE case (Eq.~\ref{eq:new_constraint3}) have seven constraints because the chosen constrained vertex removes one of the equality constraints; for example, since $\mathbf{p}_{s}=\mathbf{p}_{0}\in \mathcal{C}_{p}$,  $\mathbf{n}\cdot(\mathbf{p}_{0}- \mathbf{p}_{0})=0$.
%

\begin{figure}[htb]
      \centering\subfigure[FV or VF case] {\includegraphics[height=2.5cm]{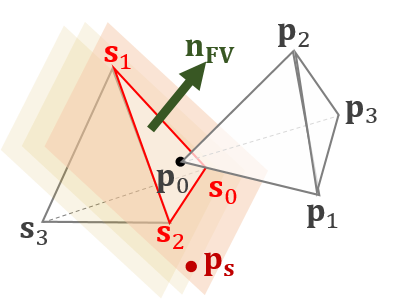}}
      \centering\subfigure[EE case]{
      \includegraphics[height=2.5cm]{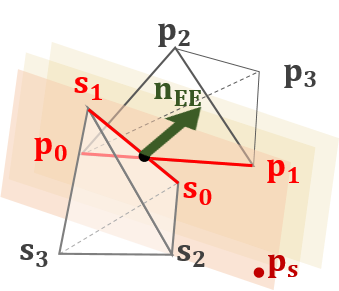}}
      \vspace{-1.5mm}
      \caption{A new variable ${\mathbf{p}_{s}}$ constrains the location of the separating plane.  Constrained vertices (red) should lie on the separating plane and free vertices (black) should lie inside either side of  the plane.}
      \vspace{-3mm}
      \label{fig:constraints}
 \end{figure}
 
\subsection{Constrained Optimization}

\vspace{-1mm}
Similarly to the previous section, $\PDd$ can be calculated by solving the QP problem of Eq.~\ref{eq:d_PD_definition}.
Note that the non-penetration constraints in Eq.~\ref{eq:constraints} is still linear since the separating directions are constant.

\vspace{-1mm}
\section{Acceleration Technique}\label{sec:case3}
\vspace{-1mm}
In order to compute the $\PDd$, we have to solve a QP problem with constraints, each of which is associated with a separation direction, typically the entire solution taking tens of milliseconds using an off-the-shelf QP solver. 
However, if we can pre-select a set of candidate directions, possibly containing the optimal separating direction, we can significantly reduce the overall computation time. Under the hypothesis that the deformation of soft objects is not severe, we assume that the penetration (equivalently separating) direction for deformable objects is close to the penetration direction when the objects behave rigidly. 
Moreover, in our experiment, we have found that, on average, 52.33\% of the result obtained by rigid penetration depth coincides with the optimal direction after running full optimization on $\PDd$.  

\vspace{-1mm}
To leverage this observation and accelerate $\PDd$ computation, we calculate the rigid penetration depth PD based on the separating axis theorem (SAT) \cite{gottschalk1996obbtree} and feed it to the optimization problem in Eq.~\ref{eq:d_PD_definition}; this is equivalent to using a single contact constraint in Eq.~\ref{eq:d_PD_definition}.
There are various methods to find rigid penetration depth such as using GJK algorithm or Minkowski sums \cite{cameron1986determining, dobkin1993computing}. However, we choose to use a SAT-based algorithm since our $\PDd$ algorithm can be considered as a general case of PD and also can be implemented similarly to \cite{lee2010simple}. 
Specifically, in order to compute PD, we use the overlap length of tetrahedra projected to each separating direction and take the projected, minimum overlap length as PD.
A little caution has to be taken in our case, however, since tetrahedra are not symmetric. Thus, we need to project all tetrahedral vertices to 44 separating direction to calculate the minimum overlapping length.

\vspace{-1mm}
Let $L_\mathbf{n}(\mathcal{T})=  \max ( \{  |(\mathbf{x}_{i}-\mathbf{x}_{j})\cdot\mathbf{n}| ~ \vert ~ \forall \mathbf{x}_{i},  \mathbf{x}_{j}\in\mathcal{T}\})$ for a tetrahedron $\mathcal{T}$, vertices $\mathbf{x}_i, \mathbf{x}_j$ in $\mathcal{T}$, and $\mathbf{n}$ is a given separating direction. The PD can be calculated as: 
\begin{equation}\label{eq:PD}
\mathrm{PD}(\Ta,\Tb)\! =\! \min(\{(L_\mathbf{n}(\Ta)\!+\!L_\mathbf{n}(\Tb)\! -\!L_\mathbf{n}(\Ta\cup\Tb)~\vert 
\forall\mathbf{n}\!\in\! N\}),
\end{equation}

\vspace{-1.5mm}
where $N=\{\mathbf{n}_\mathrm{FV},\mathbf{n}_\mathrm{VF},\mathbf{n}_\mathrm{EE}\}$ is a set of possible separating directions between $\Ta, \Tb$ (see the appendix for proof).
Interestingly, instead of taking the "minimum direction" in Eq.~\ref{eq:PD}, even if we take only the first six directions in ascending order of minimum distances, more than 90\% of $\PDd$ direction from Sec.~\ref{sec:case2} can be still found. This observation opens up a new possibility of reducing the possible error from this approximation while spending a little more time on search.

\section{RESULTS AND DISCUSSION}\label{sec:results}

Now we show implementation results of our penetration depth algorithms for static/deformable, deformable/deformable, and deformable/deformable with acceleration. 
For our experiments, 
the target tetrahedra are randomly generated, which are bounded by a cube with a side length of 10. We tested $10^4$ intersecting pairs of tetrahedra with different volume sizes ranging from 0.1 to 130, and compute their $\PDd$'s.
We implemented our algorithms using C++ on a Windows 10 PC with an AMD Rizen7 1700x 3.6GHz CPU, and 32GB memory. We used the Gurobi Optimizer to solve the QP problems.
\vspace{-1.5mm}
\subsection{Performance}
\vspace{-1.5mm}
As illustrated in Fig.~\ref{fig:implementation}, every result using $\PDd$ resolves penetration. The average performance results of each case are shown in Table.~\ref{tab:performance}. The running time for the static/deformable case is slightly faster than the deformable/deformable case since a fewer number of variables are required to optimize. 
The accelerated deformable/deformable case can be calculated in 1.07 milliseconds on average.
\vspace{-1mm}
    \begin{table}[h]
    \caption{Performance Statistics}
     \vspace{-3mm}
    \label{tab:performance}
    \begin{center}
    \begin{tabular}{|c||c|c|c|}
    \hline
    Criterion & STAT/DEF & DEF/DEF & DEF/DEF(ACCEL)\\
    \hline
    \hline
    {Performance} & {32.26\it{ms}} & {46.29 \it{ms}} & {1.07 \it{ms}}\\
    \hline
    $\PDd/\mathrm{PD}$ &{43.59 \it{\%}} &{27.94 \it{\%}}& {29.39\it{\%}}\\
    (Standard Deviation) & (13.17\it{\%}) & (3.15\it{\%}) & (3.47\it{\%})\\
    \hline
    \end{tabular}
    \end{center}
    \vspace{-6mm}
 \end{table}
 \begin{figure}[htb]
      \centering\subfigure[Penetration State] {\includegraphics[height=2.3cm]{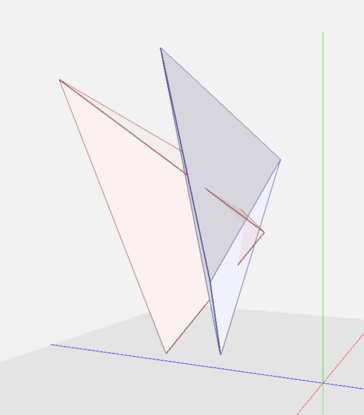}}
      \subfigure[Static / Deformable]
      {\includegraphics[height=2.3cm]{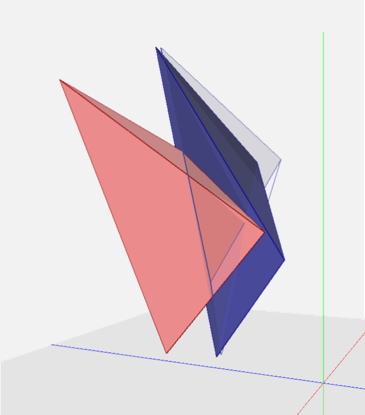}}
      \subfigure[Deformable / Deformable]
      {\includegraphics[height=2.3cm]{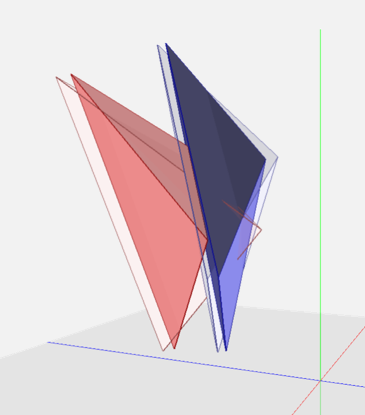}}
       \subfigure[Rigid PD]
       {\includegraphics[height=2.3cm]{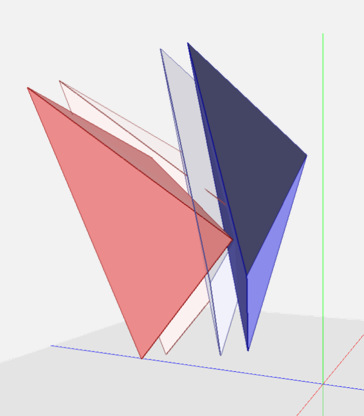}}
       \vspace{-1.5mm}
      \caption{Implementation results of deformable penetration depth and rigid penetration depth. (a) shows the initial penetrated state of two tetrahedra. In (b)$\sim$(d), the solid colored tetrahedra are the results with $\PDd = 0.5692, 0.3469, \mathrm{PD} = 1.496$, respectively.}
      \label{fig:implementation}
  \end{figure}
  \begin{figure}[htb]
  \vspace{-3mm}
      \centering
      \includegraphics[height=3.3cm]{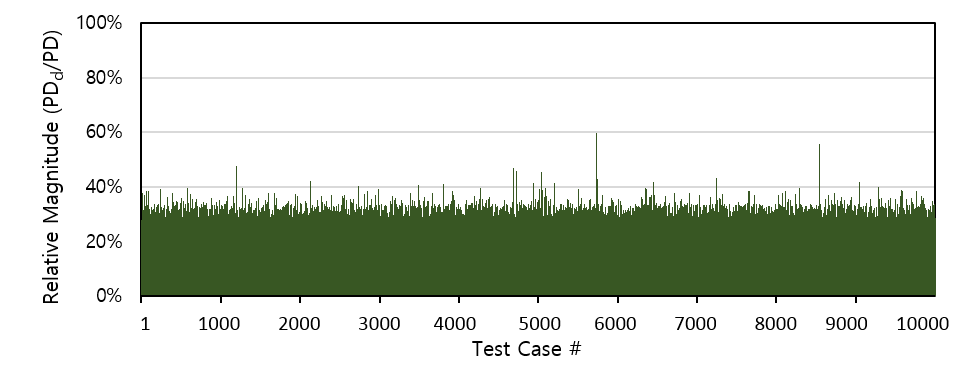}
      \vspace{-8mm}
      \caption{Relative magnitude of $\PDd$ over rigid penetration depth $\mathrm{PD}$ in $10^4$ times of penetration resolution tests of randomly intersecting tetrahedron pairs. The average magnitude is 27.94\% with the standard deviation of 3.15\%.}
      \label{fig:Relative_Manitude}
        \vspace{-6mm}
  \end{figure}
\subsection{Discussion}
\vspace{-1.5mm}
To the best of our knowledge, there are no penetration depth algorithms for deformable models available to guarantee full separation as a result (i.e. penetration metric with a tight upper bound). All the known penetration depth algorithms for deformable models provide only a lower bound, which does not guarantee full penetration resolution. Thus, it is unfair to compare our algorithms against existing algorithms for deformables and instead we compare our algorithms with rigid PD, which can be considered as an upper bound for deformable penetration depth.
In this case, in order to show the tightness of the metric upper bounds, we calculate the relative magnitude of $\PDd$ with respect to rigid penetration depth PD. 
Fig.~\ref{fig:Relative_Manitude} shows that the magnitude of $\PDd$ is less than $\frac{1}{3}$ of that of the rigid case most times. 
This means that $\PDd$ provides a much tighter deformation metric than using the rigid penetration depth. 

\section{CONCLUSION}\label{sec:conc}

 \vspace{-1mm}
We have formulated a new penetration metric based on object norm for a pair of intersecting tetrahedra undergoing linear deformation. The new metric, called deformable PD $\PDd$, optimizes an average displacement of all the points inside the tetrahedra under linear deformation to separate them. The deformable PD can be computed by solving a QP problem based on the distance metric with non-penetration constraints. 

\vspace{-1.5mm}
We have implemented three cases of computing $\PDd$, rigid vs. deformable, deformable vs. deformable with and without acceleration. Our experimental results show that we can compute $\PDd$ in a fraction of milliseconds for intersecting, deformable tetrahedra.

 \vspace{-1.5mm}
There are a few limitations to our algorithm. To derive a tractable optimization problem for $\PDd$, we have assumed that the separation direction can be obtained from the rest pose of deforming tetrahedra. Even though it is straightforwardly extendable to a set of tetrahedra by applying our metric to each element in the set, it might be interesting to pursue a technique to accelerate this computation, to be more useful for FEM-type simulation. One plausible direction would be to combine the iterative contact space projection technique \cite{PolyDepth} with our deformable metric.
 Our metric does not guarantee volume-preservation during deformation, which is another interesting future work.

\section*{Acknowledgements}
This project was supported by the National Research Foundation (NRF) in South Korea (2017R1A2B3012701).







\bibliographystyle{IEEEtran}

\bibliography{pd}

\appendix
\input{appendix.tex}

\end{document}

%% file: appendix.tex
We prove that the penetration depth (PD) between two intersected tetrahedra can be calculated using the separating axis theorem (SAT), which was claimed without proof in Eq. 11. The main idea of the proof is that the PD is equal to the smallest overlapping length projected over all possible separating axes in the SAT.
\begin{definition}
 The projected length $L_\mathbf{n}(\mathcal{T})$ of a given simplicial complex $\mathcal{T}$ along an axial direction $\mathbf{n}$ is defined as:
 \begin{equation}\label{eq:projection}
    L_\mathbf{n}(\mathcal{T})=  \max ( \{  |(\mathbf{x}_{i}-\mathbf{x}_{j})\cdot\mathbf{n}| ~ \vert ~ \forall \mathbf{x}_{i},  \mathbf{x}_{j}\in\mathcal{T}\})
 \end{equation}
 where $\mathbf{x}_i, \mathbf{x}_j$ are the vertices in $\mathcal{T}$. Then, the pair of vertices $\mathbf{x}_i, \mathbf{x}_j$ that realizes the projected length $L_\mathbf{n}(\mathcal{T})$ is called supporting vertices for $\mathbf{n}$.
\end{definition}

\begin{theorem}\label{thm:pd}
The penetration depth (PD) of two intersected tetrahedra $\Ta, \Tb$ can be calculated as:
\begin{equation}\label{eq:PD}
\mathrm{PD}(\Ta,\Tb) = \min\{(L_\mathbf{n}(\Ta)+L_\mathbf{n}(\Tb) -L_\mathbf{n}(\Ta\cup\Tb))~\vert 
\forall\mathbf{n}\in N\},
\end{equation}
where $N=\{\mathbf{n}_\mathrm{FV},\mathbf{n}_\mathrm{VF},\mathbf{n}_\mathrm{EE}\}$ is a set of possible separating directions for $\Ta, \Tb$.

\end{theorem}

\begin{proof}
According to the separating axis theorem (SAT) \cite{gottschalk1996obbtree}, two convex objects $\Ta, \Tb$ do not overlap iff there exists a separating axis $\mathbf{n}$ that the axial projection of $\Ta, \Tb$ onto $\mathbf{n}$ does not overlap.
%
The SAT can be rewritten using Eq.~\ref{eq:projection} as:
\begin{equation}\label{eq:SAT}
\exists \mathbf{n} \in N, \;  L_\mathbf{n}(\Ta \cup \Tb) \geq L_\mathbf{n}(\Ta)+ L_\mathbf{n}(\Tb).
\end{equation}
%
Thus, if two objects are interpenetrated, $ L_\mathbf{n}(\Ta)+ L_\mathbf{n}(\Tb) - L_\mathbf{n}(\Ta \cup \Tb) > 0 $ for $\forall \mathbf{n}$. Let $\varepsilon$ be the result of evaluating Eq.~\ref{eq:PD} and $\mathbf{m}$ be the corresponding separating direction: i.e.,
\begin{equation}\label{eq:separatingDirection}
 \mathbf{m} = \argmin_{\mathbf{n}}\{(L_\mathbf{n}(\Ta)+L_\mathbf{n}(\Tb) -L_\mathbf{n}(\Ta\cup\Tb))~\vert 
\forall\mathbf{n}\in N\}
\end{equation}
Then, we can prove Theorem \ref{thm:pd} by showing that:
\begin{enumerate}
    \item $\varepsilon\mathbf{m}$ separates $\Ta$ and $\Tb$ by translation.
    \item $\varepsilon$ is the smallest magnitude among such translations.
\end{enumerate}

\begin{figure}[htb]
      \centering\subfigure[Interpenetrated state] {\includegraphics[height=3cm]{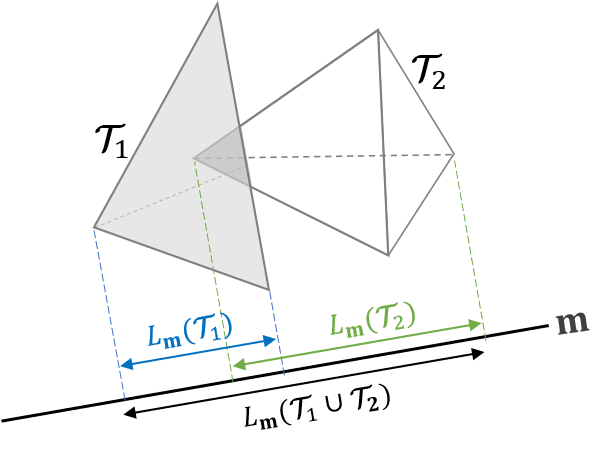}}
       \centering\subfigure[Translated state] {\includegraphics[height=3cm]{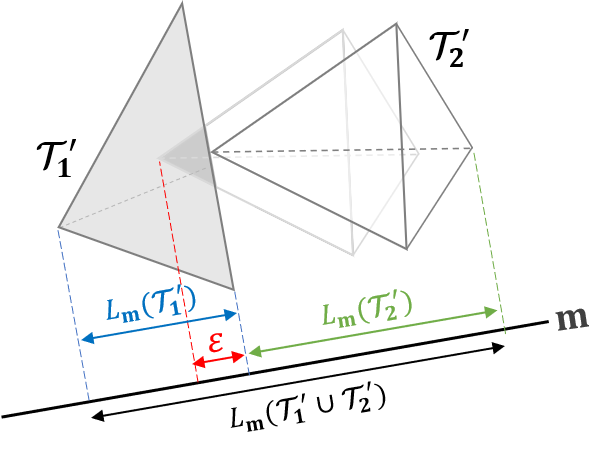}}
      \caption{Example of two tetrahedra projected on the axis $\mathbf{m}$ before and after translation $\varepsilon \mathbf{m}$, where $\mathbf{m}$ is the face normal of $\Ta$. The projected length of each tetrahedron does not change during translation. }
      \label{fig:projection}
  \end{figure}

Let $\Ta', \Tb'$ be the tetrahedra of $\Ta, \Tb$ translated by $\varepsilon\mathbf{m}$. Since these tetrahedra are not rotated, the projected length of each tetrahedron on axis $\mathbf{m}$ is the same as before translation:
\begin{equation}
\begin{split}
L_\mathbf{m}(\Ta)= L_\mathbf{m}(\Ta')\\
L_\mathbf{m}(\Tb)= L_\mathbf{m}(\Tb')
\end{split}
\end{equation}

Since the direction of the translation is the same as the projection axis, the length of the translation vector is  $\vert\varepsilon\mathbf{m}\cdot\mathbf{m}\vert=\varepsilon$. Then, the entire projected length of the translated tetrahedra is:
\begin{equation}
\begin{split}
& L_\mathbf{m}(\Ta' \cup \Tb')\\
& = L_\mathbf{m}(\Ta \cup \Tb) + \varepsilon\\ 
& = L_\mathbf{m}(\Ta \cup \Tb) + L_\mathbf{m}(\Ta)+ L_\mathbf{m}(\Tb) - L_\mathbf{m}(\Ta \cup \Tb)\\
& = L_\mathbf{m}(\Ta)+ L_\mathbf{m}(\Tb)\\
& = L_\mathbf{m}(\Ta')+ L_\mathbf{m}(\Tb'),
\end{split}
\end{equation}
implying that the two tetrahedra translated by $\varepsilon\mathbf{m}$ do not overlap because of the SAT (Fig.~\ref{fig:projection}(b)). 

Let $\tilde{\varepsilon}\tilde{\mathbf{m}}$ be an arbitrary translation that separates $\Ta, \Tb$ and $\tilde{\Ta},\tilde{\Tb}$ be the translated copies of the tetrahedra. Since $\tilde{\Ta},\tilde{\Tb}$ are separated, according to the SAT, there exists a separating axis $\tilde{\mathbf{n}}$ that satisfies $L_\mathbf{\tilde{n}}(\tilde{\Ta} \cup \tilde{\Tb}) \geq L_\mathbf{\tilde{n}}(\tilde{\Ta})+ L_\mathbf{\tilde{n}}(\tilde{\Tb})$. 
For the given set $\Ta \cup \Tb$ and the direction $\tilde{\mathbf{n}}$, let $\mathbf{x}_1$ and $\mathbf{x}_2$ be the two supporting vertices used to calculate the projection $L_\mathbf{\tilde{\mathbf{n}}}(\Ta \cup \Tb) = |(\mathbf{x}_2-\mathbf{x}_1)\cdot\tilde{\mathbf{n}}|$. 

Now, suppose that these two vertices can support the tetrahedra even after being translated by $\tilde{\varepsilon}\mathbf{\tilde{m}}$ (Fig.~\ref{fig:arbitraryTranslation}(a)), and let $\mathbf{\tilde{x}_1}$ and $\mathbf{\tilde{x}_2}$ be the corresponding vertices after translation. Then the displacement between the two vertices after translation can be calculated as: $\mathbf{\tilde{x}_2}-\mathbf{\tilde{x}_1}=\mathbf{x}_2-\mathbf{x}_1+\varepsilon\mathbf{\tilde{m}}$. 
The projected length is  $L_\mathbf{\tilde{\mathbf{n}}}(\tilde{\Ta} \cup \tilde{\Tb}) 
=|(\mathbf{\tilde{x}_2}-\mathbf{\tilde{x}_1})\cdot\tilde{\mathbf{n}}|= |(\mathbf{x}_2-\mathbf{x}_1+\varepsilon\mathbf{\tilde{m}})\cdot\tilde{\mathbf{n}}|$. 
Then,
\begin{equation}
\begin{split}
& L_\mathbf{\tilde{n}}(\Ta \cup \Tb) + \vert\tilde{\varepsilon}\mathbf{\tilde{m}}\cdot\mathbf{\tilde{n}}\vert\\
&=|(\mathbf{x}_2-\mathbf{x}_1)\cdot\tilde{\mathbf{n}}|+|\varepsilon\mathbf{\tilde{m}}\cdot\tilde{\mathbf{n}}|\\
&\geq |(\mathbf{x}_2-\mathbf{x}_1)\cdot\tilde{\mathbf{n}}+\varepsilon\mathbf{\tilde{m}}\cdot\tilde{\mathbf{n}}| 
= L_\mathbf{\tilde{n}}(\tilde{\Ta} \cup \tilde{\Tb})\\
&\geq L_\mathbf{\tilde{n}}(\tilde{\Ta})+ L_\mathbf{\tilde{n}}(\tilde{\Tb}) = L_\mathbf{\tilde{n}}(\Ta)+ L_\mathbf{\tilde{n}}(\Tb)
\end{split}
\end{equation}

\begin{figure}[htb]
      \subfigure[Supporting vertices are same] {\includegraphics[height=3cm]{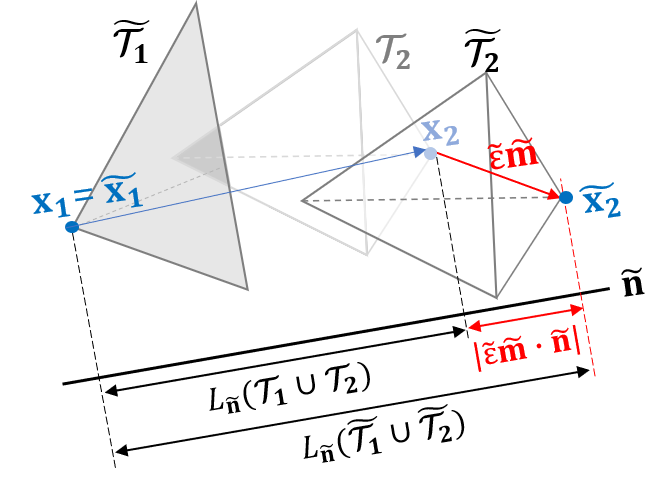}}
       \subfigure[Supporting vertices are changed] {\includegraphics[height=3cm]{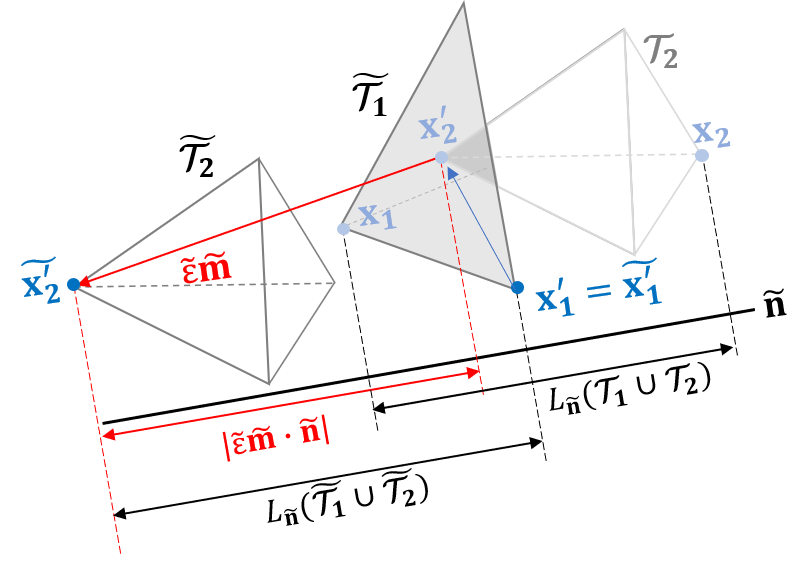}}
      \caption{Examples of an arbitrary translation $\tilde{\varepsilon}\tilde{\mathbf{m}}$ that separates two tetrahedra and their projection on the separating axis $\mathbf{\tilde{n}}$. (a)(b) shows two cases of projection results according to the changes in supporting vertices before and after translation. }
      \label{fig:arbitraryTranslation}
  \end{figure}
  
Otherwise, the supporting vertices are changed after translation (Fig.~\ref{fig:arbitraryTranslation}(b)). Let $\tilde{\mathbf{x}_1}'$, $\tilde{\mathbf{x}_2}'$ be the supporting vertices after translation and $\mathbf{x}_1'$, $\mathbf{x}_2'$ be the corresponding vertices before translation. Then, $L_\mathbf{\tilde{\mathbf{n}}}(\tilde{\Ta} \cup \tilde{\Tb}) 
=|(\mathbf{\tilde{x}_2}'-\mathbf{\tilde{x}_1}')\cdot\tilde{\mathbf{n}}|= |(\mathbf{x}_2'-\mathbf{x}_1'+\varepsilon\mathbf{\tilde{m}})\cdot\tilde{\mathbf{n}}|$. Since $\mathbf{x}_1'$ and $\mathbf{x}_2'$ were not supporting vertices before translation, according to Eq.~\ref{eq:projection}, the projected length of  $\mathbf{x}_1'$ and $\mathbf{x}_2'$ must be smaller than that of supporting vertices $\mathbf{x}_1$, $\mathbf{x}_2$: i.e.,
$L_\mathbf{\tilde{n}}(\Ta \cup \Tb)=|(\mathbf{x}_2-\mathbf{x}_1)\cdot\tilde{\mathbf{n}}|\geq|(\mathbf{x}_2'-\mathbf{x}_1')\cdot\tilde{\mathbf{n}}|$. Thus,
\begin{equation}
\begin{split}
& L_\mathbf{\tilde{n}}(\Ta \cup \Tb) + \vert\tilde{\varepsilon}\mathbf{\tilde{m}}\cdot\mathbf{\tilde{n}}\vert\\
&=|(\mathbf{x}_2-\mathbf{x}_1)\cdot\tilde{\mathbf{n}}|+|\varepsilon\mathbf{\tilde{m}}\cdot\tilde{\mathbf{n}}|\\
&\geq |(\mathbf{x}_2'-\mathbf{x}_1')\cdot\tilde{\mathbf{n}}|+|\varepsilon\mathbf{\tilde{m}}\cdot\tilde{\mathbf{n}}|\\
&\geq |(\mathbf{x}_2'-\mathbf{x}_1')\cdot\tilde{\mathbf{n}}+\varepsilon\mathbf{\tilde{m}}\cdot\tilde{\mathbf{n}}| = L_\mathbf{\tilde{n}}(\tilde{\Ta} \cup \tilde{\Tb})\\
&\geq L_\mathbf{\tilde{n}}(\tilde{\Ta})+ L_\mathbf{\tilde{n}}(\tilde{\Tb}) = L_\mathbf{\tilde{n}}(\Ta)+ L_\mathbf{\tilde{n}}(\Tb)
\end{split}
\end{equation}

In either case, we can see that $\vert\tilde{\varepsilon}\mathbf{\tilde{m}}\cdot\mathbf{\tilde{n}}\vert \ge  L_\mathbf{\tilde{n}}(\Ta)+ L_\mathbf{\tilde{n}}(\Tb) - L_\mathbf{\tilde{n}}(\Ta \cup \Tb)$.
Since $|\tilde{m}| =|\tilde{n}| = 1$, 
\begin{equation}
\begin{split}
\tilde{\varepsilon} &\geq \tilde{\varepsilon}\vert\mathbf{\tilde{m}}\cdot\mathbf{\tilde{n}}\vert \\
&\geq L_\mathbf{\tilde{n}}(\Ta)+ L_\mathbf{\tilde{n}}(\Tb) - L_\mathbf{\tilde{n}}(\Ta \cup \Tb)\\
&\geq \varepsilon
\end{split}
\end{equation}
Therefore, $\varepsilon$ is the minimum translational distance that separates $\Ta, \Tb$.
\end{proof}